\documentclass[10pt, conference, a4paper]{IEEEtran}
\usepackage{amsmath,amsfonts,amssymb,amsbsy, amsthm,ae,aecompl}
\usepackage{algorithm,algpseudocode}
\usepackage[utf8]{inputenc}
\usepackage[english]{babel}
\usepackage{bm}
\usepackage{color}
\usepackage[noadjust]{cite}
\usepackage{epsfig}
\usepackage{enumerate}
\usepackage{float}
\usepackage{fancyhdr}
\usepackage[T1]{fontenc}
\usepackage[acronym,toc,shortcuts]{glossaries}
\usepackage{graphicx, caption, subcaption}
\usepackage{hyperref}
\usepackage{bbm}
\usepackage{lastpage}
\usepackage{listings}
\usepackage{lipsum}
\usepackage{dsfont}
\usepackage{multind}
\usepackage{multirow,tabularx}
\usepackage[normalem]{ulem}
\usepackage{transparent}
\usepackage{tikz,pgfplots}
\usepackage{times}
\usepackage{verbatim}
\usepackage[all]{xy}
\hypersetup{
    colorlinks,%
    citecolor=black,%
    filecolor=black,%
    linkcolor=black,%
    urlcolor=black
}

\usetikzlibrary{calc}
\usetikzlibrary{arrows,decorations.markings}
\pgfplotsset{
  grid style = {
    dash pattern = on 0.025mm off 0.95mm on 0.025mm off 0mm, % start with half a dot to get correct centering of the pattern
    line cap = round,
    black,
    line width = 0.5pt
  },
  tick label style={font=\small},
  label style={font=\small},
  legend style={font=\footnotesize},
}

\pgfplotscreateplotcyclelist{laneas-delay1}{
	cyan!80!black, solid, thick, mark=none\\%						% Theory 1
	red!80!black, solid, thick, mark=none\\%							% Theory 2	
	lime!40!black, solid, very thick, mark=none\\%					% Theory 3	
	lime!60!black, solid, very thick, mark=none\\%					% Theory 4	
	lime!80!black, solid, very thick, mark=none\\%					% Theory 5	
	cyan!80!black, solid, thick, only marks, mark=+, mark size=2, every mark/.append style={solid,fill=cyan!80!black}\\%									% Simulation 1	
	red!80!black, solid, thick, only marks, mark=star, mark size=2, every mark/.append style={solid,fill=red!80!black}\\%									% Simulation 2
	lime!40!black, solid, thick, only marks, mark=square, mark size=2, every mark/.append style={solid,fill=lime!80!black}\\%									% Simulation 3
	lime!60!black, solid, thick, only marks, mark=o, mark size=2, every mark/.append style={solid,fill=lime!80!black}\\									% Simulation 4
	lime!80!black, solid, thick, only marks, mark=diamond, mark size=2, every mark/.append style={solid,fill=lime!80!black}\\%									% Simulation 5
}

\graphicspath{{figures/}}

\bgroup
\makeglossaries
\newacronym{ADMM}{ADMM}{Alternating Direction Method of Multipliers}
\newacronym{APC}{APC}{area power consumption}
\newacronym{ASE}{ASE}{area spectral efficiency}
\newacronym{BS}{BS}{base station}
\newacronym{CDF}{CDF}{cumulative distribution function}
\newacronym{CN}{CN}{core network}
\newacronym{CS}{CS}{central scheduler}
\newacronym{CR}{CR}{central router}
\newacronym{D2D}{D2D}{device-to-device}
\newacronym{iid}{i.i.d.}{independent and identically distributed}
\newacronym{EE}{EE}{energy efficiency}
\newacronym{FKG}{FKG}{{F}ortuin-{K}asteleyn-{G}inibre}
\newacronym{LTE}{LTE}{long term evolution}
\newacronym{LRU}{LRU}{least-recently used}
\newacronym{LFU}{LFU}{least-frequently used}
\newacronym{massive-MIMO}{massive-MIMO}{massive multiple-input multiple-output}
\newacronym{MBS}{MBS}{macro base station}
\newacronym{MIMO}{MIMO}{multiple-input multiple-output}
\newacronym{MU}{MU}{macro cell user}
\newacronym{PP}{PP}{point process}
\newacronym{PPP}{PPP}{{P}oisson point process}
\newacronym{RAN}{RAN}{radio access network}
\newacronym{SBS}{SBS}{small base station}
\newacronym{SNR}{SNR}{signal-to-noise ratio}
\newacronym{SINR}{SINR}{signal-to-interference-plus-noise ratio}
\newacronym{SIR}{SIR}{signal-to-interference ratio}
\newacronym{SCN}{SCN}{small cell network}
\newacronym{SSD}{SSD}{solid-state disk}
\newacronym{SU}{SU}{small cell user}
\newacronym{UT}{UT}{user terminal}
\newacronym{TDMA}{TDMA}{time-division multiple access}
\newacronym{QoS}{QoS}{quality-of-service}
\newacronym{QoE}{QoE}{quality-of-experience}
\newacronym{PDF}{PDF}{probability distribution function}
\newacronym{PGFL}{PGFL}{probability generating functional}
\newacronym{RHS}{RHS}{right hand side}
\newacronym{HetNet}{HetNet}{heterogeneous network}

\newtheorem{theorem}{Theorem}

\newtheorem{corollary}{Corollary}

\begin{document} 
\title{On the Delay of Geographical Caching Methods in Two-Tiered Heterogeneous Networks}
\author{
		\IEEEauthorblockN{Ejder Baştuğ$^{\diamond}$, Marios Kountouris$^{\circ}$, Mehdi Bennis$^{\dagger}$, and Mérouane Debbah$^{\diamond, \circ}$\\ \vspace{-0.3cm}}
		\IEEEauthorblockA{
				\small
				$^{\diamond}$Large Networks and Systems Group (LANEAS), CentraleSupélec, \\ Université Paris-Saclay, 3 rue Joliot-Curie,  91192 Gif-sur-Yvette, France \\	
				% $^{\star}$Department of Telecommunications, CentraleSupélec, 91192, Gif-sur-Yvette, France \\
				$^{\circ}$Mathematical and Algorithmic Sciences Lab, Huawei France R\&D, Paris, France \\	
				$^{\dagger}$Centre for Wireless Communications, University of Oulu, Finland \\	
				\{ejder.bastug, merouane.debbah\}@centralesupelec.fr, marios.kountouris@huawei.com, bennis@ee.oulu.fi 	
				\vspace{-0.60cm}
		}
		\thanks{This research has been supported by the ERC Starting Grant 305123 MORE (Advanced Mathematical Tools for Complex Network Engineering), the projects BESTCOM and 4GinVitro, the Academy of Finland CARMA project and TEKES grant (2364/31/2014).}
}
\IEEEoverridecommandlockouts
\maketitle
%\IEEEpeerreviewmaketitle
%%%%%%%%%%%%%%%%%%%%%%%%%%%%%%%%%%%%%%%%%%%%%%%%%%%%%%%%%%%%%%%%%%%
%%%%%%%%%%%%%%%%%%%%%%%%%%%%%%%%%% ABSTRACT
\begin{abstract}
We consider a hierarchical network that consists of mobile users, a two-tiered cellular network (namely small cells and macro cells) and central routers, each of which follows a Poisson point process (PPP). In this scenario, small cells with limited-capacity backhaul are able to cache content under a given set of randomized caching policies and storage constraints. Moreover, we consider three different content popularity models, namely \emph{fixed} content popularity, \emph{distance-dependent} and \emph{load-dependent}, in order to model the spatio-temporal behavior of users' content request patterns. We derive expressions for the average delay of users assuming perfect knowledge of content popularity distributions and randomized caching policies. Although the trend of the average delay for all three content popularity models is essentially identical, our results show that the overall performance of cached-enabled heterogeneous networks can be substantially improved, especially under the load-dependent content popularity model.
\end{abstract}
%%%%%%%%%%%%%%%%%%%%%%%%%%%%%%%%%%%%%%%%%%%%%%%%%%%%%%%%%%%%%%%%%%%
%%%%%%%%%%%%%%%%%%%%%%%%%%%%%%%%%% INTRODUCTION
\section{Introduction}
Content caching in $5$G heterogeneous wireless networks improves the system performance, and is of high importance in limited-backhaul scenarios \cite{Bastug2014LivingOnTheEdge}. Most existing literature for cache-enabled heterogeneous networks using stochastic geometry focuses on the characterization of key performance metrics neglecting the backhaul limitations and the spatio-temporal content popularity profiles \cite{Chen2016Cooperative, Afshang2015Modeling, Serbetci2016Cost, Yan2016User}. In order to capture these aspects, we analyze in this paper the gains of caching in heterogeneous network deployment considering the average delay as a performance metric.

Consider a multi-tier heterogeneous network where base stations in each tier are deployed according to a homogeneous \ac{PPP}. More precisely, we model a heterogeneous network which consists of mobile terminals (users), cache-enabled \glspl{SBS}, \glspl{MBS} and central routers. In this network setting, a user may experience delays due to  downlink transmissions, backhaul and caches. Supposing that \glspl{SBS} are able to cache contents proactively, we derive expressions for the average delay of typical users when connected to either \glspl{MBS} or \glspl{SBS}. Moreover, in order to capture the spatio-temporal content access patterns of users, we suppose \emph{fixed} content popularity, \emph{distance-dependent} and \emph{load-dependent} content popularities. Assuming that the content popularity distribution is perfectly known at the small base stations, we explore three different caching policies based on content-popularity and randomization. 
%In the final part of this work, we validate our results via Monte-Carlo simulations and draw our conclusions.

%%%%%%%%%%%%%%%%%%%%%%%%%%%%%%%%%%%%%%%%%%%%%%%%%%%%%%%%%%%%%%%%%%%%
%%%%%%%%%%%%%%%%%%%%%%%%%%%%%%%%%%% SYSTEM MODEL
\section{System Model}\label{sec:systemmodel}
%%%%%%%%%%%%%%%%%%%%%%%%%%%%%%%%%%%%%%%%%%%%%%%%%%%%%%%%%%%%%%%%%%%%
%%%%%%%%%%%%%%%%%%%%%%%%%%%%%%%%%%% Topology
\emph{Topology}: We consider a multi-tier heterogeneous network in the two-dimensional Euclidean plane $\mathbb{R}^{2}$, where nodes in each tier $k$ are distributed according to a homogeneous \ac{PPP} $\Phi_{k} = \{r^{(k)}_{i}\}_{i \in \mathbb{N}}$ of intensity $\lambda_k$, and $r^{(k)}_{i} \in \mathbb{R}^{2}$ represents the location of the $i$-th node at the $k$-th tier. The above network layout models a multi-tier heterogeneous network that consists of mobile terminals (users), \glspl{SBS}, \glspl{MBS} and central routers with densities $\lambda_{\mathrm{ut}} > \lambda_{\mathrm{sc}} > \lambda_{\mathrm{mc}} > \lambda_{\mathrm{cr}}$, respectively. A \emph{typical} mobile user is assumed to be located at the Cartesian origin $(0,0)$ in order to derive the performance metrics of the heterogeneous network.

%%%%%%%%%%%%%%%%%%%%%%%%%%%%%%%%%%%%%%%%%%%%%%%%%%%%%%%%%%%%%%%%%%%%
%%%%%%%%%%%%%%%%%%%%%%%%%%%%%%%%%%% Signal Model
\emph{Signal Model}: We shall consider that the \glspl{MBS} and \glspl{SBS} are transmitting in the same frequency band and hence interfering with each other. The transmit power is $P_{\mathrm{mc}}$ for each \ac{MBS} and $P_{\mathrm{sc}}$ for each \ac{SBS}, where we assume that $P_{\mathrm{mc}} > P_{\mathrm{sc}}$. For notational convenience, let us denote a base station (transmitter) by its position. The received power experienced at a typical user due to a transmitter $x$ is given by $P_x h_x \ell(x)$, where $P_{x}$ is the transmit power ($P_{\mathrm{mc}}$ or $P_{\mathrm{sc}}$), $h_x$ corresponds to the fast fading power coefficient (square of the fading amplitude) of the channel between transmitter $x$ and typical user, and $\ell(x) = \Vert x \Vert^{-\alpha}$ is the standard power law pathloss function with $\alpha > 2$. The channel fading power coefficients are \ac{iid} exponential random variables (Rayleigh fading) with $\mathbb{E}[h_x] = 1$.

Since we assume that the network is interference-limited (i.e., the interference power dominates over the noise power), we simply consider the \ac{SIR}. For a typical user connected to a \ac{MBS} located at $x$, the \ac{SIR} is given as
		\begin{equation}
			\mathrm{SIR}_{\mathrm{mc}}(x) = 
			\frac{ P_{\mathrm{mc}} h_x \ell(x)}
			{
				I_{\mathrm{mm}} + I_{\mathrm{sm}}
			}  \label{eq:sirMacro}
		\end{equation}
		where $I_{\mathrm{mm}} = \sum\limits_{y \in \Phi_{\mathrm{mc}} \setminus \{x\}} P_{\mathrm{mc}} h_y \ell(y)$ is the interference experienced from all \glspl{MBS} except the serving \ac{MBS} at $x$, and $I_{\mathrm{sm}} = \sum\limits_{y \in \Phi_{\mathrm{sc}}} P_{\mathrm{sc}} h_y \ell(y)$ is the aggregate interference experienced from \glspl{SBS}. For a typical user connected to a \ac{SBS} located at $x$, the \ac{SIR} is given as		\begin{equation}
			\mathrm{SIR}_{\mathrm{sc}}(x) = 
			\frac{ P_{\mathrm{sc}} h_x \ell(x)}
			{
				I_{\mathrm{ss}} + I_{\mathrm{ms}}
			} \label{eq:sirSmall}
		\end{equation}	
		where $I_{\mathrm{ss}} = \sum\limits_{y \in \Phi_{\mathrm{sc}} \setminus \{x\}} P_{\mathrm{sc}} h_y \ell(y)$ is the interference experienced from all \glspl{SBS} except the serving \ac{SBS}, and $I_{\mathrm{ms}} = \sum\limits_{y \in \Phi_{\mathrm{mc}}} P_{\mathrm{mc}} h_y \ell(y)$ is the aggregate interference from \glspl{MBS}. The target \ac{SIR} in our system model is denoted by $\gamma$.

%%%%%%%%%%%%%%%%%%%%%%%%%%%%%%%%%%%%%%%%%%%%%%%%%%%%%%%%%%%%%%%%%%%%
%%%%%%%%%%%%%%%%%%%%%%%%%%%%%%%%%%% Connectivity and Backhaul
\emph{Connectivity and Backhaul}: Mobile user terminals are associated with the closest base station, either \ac{SBS} or \ac{MBS}. As alluded to earlier, each \ac{MBS} or \ac{SBS} is also connected to its nearest central router. Each central router has a high-rate broadband Internet connection. The wired backhaul is used to provide this broadband connection to \glspl{MBS} and \glspl{SBS} via backaul links, such that users' requests can be satisfied. Supposing that a content request is generated by a user, the base station is then in charge of starting immediately its distribution. An illustration of this heterogeneous network under limited-capacity backhaul is illustrated in Fig. \ref{fig:scenario:delay}. 
%In the following, we detail our caching model.
%
%%%%%%%%%%%%%%%%%%%%%%%%%%%%%%%%%%%%%%%%%%%%%%%%%%%%%%%%%%%%%%%%%%%%
%%%%%%%%%%%%%%%%%%%%%%%%%%%%%%%%%%% Caching
\subsection{Caching Model}\label{sec:caching}
When a user has a content request, we assume that the request is drawn from the distribution $f_{\mathrm{pop}}$, which is in decreasing order of content popularities. More formally, the content popularity distribution of a user is a right continuous and monotonically decreasing \ac{PDF}, given by \cite{Newman2005Power}
\begin{equation}\label{eq:contentpdf}
f_{\mathrm{pop}}\left(f,\eta\right)
	=
	\begin{cases}
	\left(\eta - 1\right)f^{-\eta},
		& f \geq 1, \\
		0,			
		& f < 1,
	\end{cases}
\end{equation}
where $f$ indicates a point in the support of the corresponding content, and $\eta > 1$ parametrizes the steepness of the popularity distribution curve. 

In fact, higher values of $\eta$ results in steeper distribution, which in turn means that certain contents are highly popular than the rest of contents in $f_{\mathrm{pop}}\left(f,\eta\right)$. Conversely, lower values of $\eta$ yield a more uniform distribution, which in turns say that almost all contents have similar popularities. The content popularity of a user may be evolving over time and space, influenced by the choice of other users, and can be partially known at the base stations. This is somewhat equivalent to say that the parameter $\eta$ can take different values depending on the scenario. In our case, each base station \emph{perfectly} observes the  content popularities according to three different models as follows:
		\begin{itemize}
			\item[-] \emph{Fixed}: The content popularity is identical for all users, with fixed steepness factor of $\eta = \eta_0$. Therefore, all \glspl{SBS} observe the same distribution given by $f_{\mathrm{pop}}\left(f,\eta_0\right)$.
			\item[-] \emph{Distance-dependent}: The users have different content popularity distributions, each of them having a distance-dependent steepness factor $\eta = r$,  where $r$ is the (random) distance between a user and its serving \ac{SBS}. Therefore, we assume that each \ac{SBS} observes on average a content popularity distribution given by $f_{\mathrm{pop}}\left(f,\bar{r} \right)$, where $\bar{r}$ is the average distance between the \ac{SBS} and its users. This model is used to mimic the behavior of content popularity based on the distance (i.e., flat distribution in short distances).
			
			\item[-] \emph{Load-dependent}: The content popularity of users is load-dependent on average, each of \ac{SBS} having parameter $\eta = \lambda_{\mathrm{ut}} / \lambda_{\mathrm{sc}}$. Therefore, all \glspl{SBS} observe the content popularity distribution given by $f_{\mathrm{pop}}\left(f,\lambda_{\mathrm{ut}} / \lambda_{\mathrm{sc}} \right)$. This model is used to mimic the behavior of content popularity based on the load (i.e., steep distribution in heavy loads).
		\end{itemize}

Note that the choice of such a continuous content distribution is in fact for ease of analysis. When practical issues or analytical tractability are not a priority, Zipf-like discrete power laws can also be considered for modeling \cite{Newman2005Power}.  Indeed, content access statistics in cache-enabled web proxies \cite{Breslau1999WebZipf}, or more relevantly in base stations \cite{Shafiq2011Characterizing, Dehghan2015Complexity} are characterized by such discrete power laws (or arguably distributions).
	
For the (some of) caching policies described below, we shall assume that the content popularity distribution $f_{\mathrm{pop}}\left(f,\eta\right)$ is perfectly known at the base stations. Practically, in order to have partial knowledge of $f_{\mathrm{pop}}\left(f,\eta\right)$ for the caching policies, statistical estimation methods can be employed either at the base stations in a distributed manner or alternatively at the central routers, by using statistical tools from machine learning (see \cite{Pachos2016Technical, ElBamby2014ContentAware} for relevant discussions).
	
Given $f_{\mathrm{pop}}\left(f,\eta\right)$, the content in the interval $[1, f_0)$ is the \emph{cacheable} content and is called as \emph{catalogue}, whereas the remaining part $[f_0, \infty]$ is considered as non-cacheable content  (i.e., sensor data, voice streaming and online gaming). An interval $[f, f + \Delta f)$ in the support of $f_{\mathrm{pop}}\left(f,\eta\right)$ is dedicated to represent the probability of the $f$-th content.
Each \ac{SBS} has a storage capacity of $S$, thus it caches contents according to a given caching policy. Having such a request behavior described above and caching capabilities at the \glspl{SBS}, we consider the following \emph{offline} caching policies:
\begin{itemize}
	\item[-] \emph{StdPop} \cite{Bastug2014CacheEnabledExtended}: The most popular content from the catalogue is stored in the cache of \glspl{SBS} and requires $S_{\mathrm{p}} \geq 0$ amount of storage. We additionally assume that the track of content popularity in a \ac{SBS} requires $S_{\mathrm{0}}$ amount of storage, defined as a function of the number of contents in the catalogue and the type of algorithm employed for content popularity estimation, thereby it holds that $S = S_{\mathrm{p}} + S_{\mathrm{0}} \geq 0 $.
	
	\item[-] \emph{UniRand} \cite{Blaszczyszyn2014Geographic}: The $S$ amount of contents are cached uniformly at random. Note that this policy is not aware of the content catalogue, therefore it does not require any memory to track the content popularity profile.
	
	\item[-] \emph{MixPop}: The $S_{\mathrm{p}}$ amount of storage is used to cache the most popular content deterministically. The storage overhead is $S_0 \geq 0$ and again defined as a function of number of content and the employed algorithm. In addition, we suppose that $S_{\mathrm{u}} \geq 0$ amount of storage is used to cache content uniformly at random, thus $S = S_{\mathrm{p}} + S_{\mathrm{0}} + S_{\mathrm{u}} \geq 0$.
\end{itemize}

\begin{figure}[ht!]
	\centering
	\includegraphics[width=0.8\linewidth]{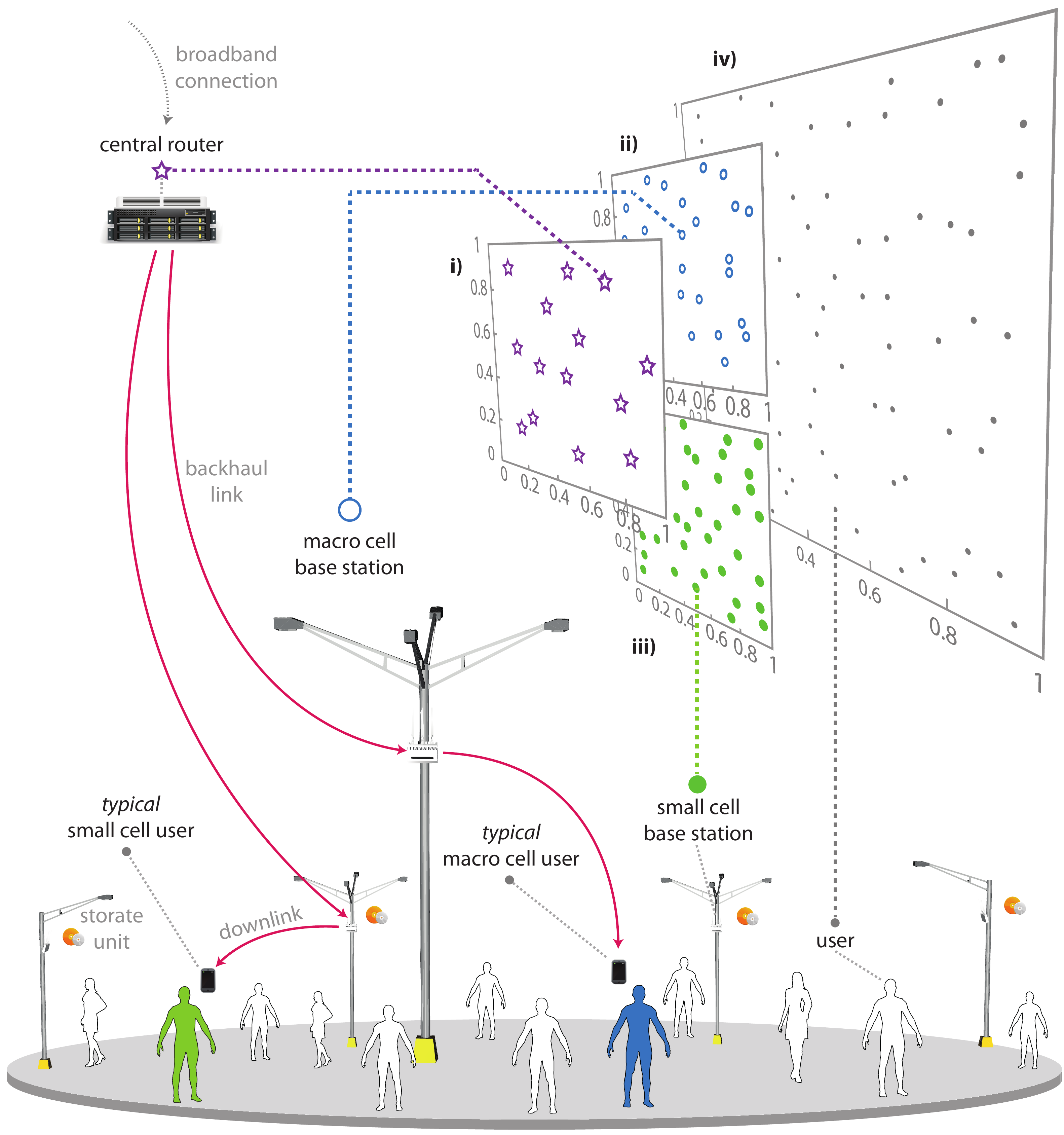}
	\caption{An illustration of the considered system model. The snapshots of i) central routers, ii) \glspl{MBS}, iii) \glspl{SBS} and iv) mobile user terminals are provided on the right side of figure.}
	\label{fig:scenario:delay}
\end{figure} 

In fact, if the catalogue size is sufficiently small, the storage overhead in StdPop and MixPop, due to the track of content popularity can be neglected. However, such an overhead may dominate the total storage space when a large catalogue with low-sized chunks is considered. One can also observe that the StdPop and UniRand policies are special cases of the MixPop policy and are given here for the sake of exposition. 

The performance of any statistics-aware \emph{online} cache removal policy (i.e., \ac{LRU} and \ac{LFU}) would be upper bounded by its offline successor that has perfect content statistics; as such an online approach would require iterative estimation of content popularity in a finite time window, yielding to overall performance degradation. Such online policies can also be incorporated to our system model after some specific assumptions (see Independent Reference Model \cite{Che2002Hierarchical} for instance).
%%%%%%%%%%%%%%%%%%%%%%%%%%%%%%%%%%%%%%%%%%%%%%%%%%%%%%%%%%%%%%%%%%%%
%%%%%%%%%%%%%%%%%%%%%%%%%%%%%%%%%%% Delay and Quality of Service
\subsection{Delay and Quality of Service}\label{sec:delayqos}
\Ac{QoS} is closely related to the delay experienced by users. We consider three different sources of delay which are detailed separately as follows.

%%%%%%%%%%%%%%%%%%%%%%%%%%%%%%%%%%% Delay in downlink
\emph{Delay in downlink}: When \glspl{MBS} and \glspl{SBS} have to deliver the contents to their intended mobile users, it is evident that the downlink transmission over the wireless medium incurs a delay mainly due to the interference from concurrent transmissions and channel fading. Consider now a simple retransmission protocol where a packet of requested content is repeatedly transmitted until its successful delivery, up to a pre-defined number of retransmission attempts $M$. Indeed, inferring  whether a packet delivery is successful or not at the base station essentially relies on the \ac{SINR} (or \ac{SIR} in our case) being higher than the predefined threshold $\gamma$.  If a packet is delivered successfully, we shall assume that the base station (macro or small cell) receives a one-bit acknowledgement message from the mobile user with negligible delay and error. Otherwise, if the delivery fails, the base station receives a one-bit negative acknowledgement message in the same vein. These attempts take $T_0$ amount of time. An \emph{outage} event occurs if the packet is not delivered after $M$ attempts. In the remainder, we denote the downlink delay experienced by the typical \glspl{MU} and \glspl{SU} as $D_{\mathrm{dm}}$ and $D_{\mathrm{ds}}$ respectively.

%%%%%%%%%%%%%%%%%%%%%%%%%%%%%%%%%%% Delay in backhaul
\emph{Delay in backhaul}: The delay caused in a wired backhaul link is modeled by an exponentially distributed random variable whose mean being proportional to the product of the average link distance (from typical base station to its nearest central router) and the average number of base stations connected to a single central router. In particular, representing the delay in macro and small cell backhaul links as $D_{\mathrm{bm}} \sim \mathrm{Exponential}(\bar{\mu}_{\mathrm{bm}})$ and $D_{\mathrm{bs}} \sim \mathrm{Exponential}(\bar{\mu}_{\mathrm{bs}})$ respectively, we (in general) suppose that $D_{\mathrm{bs}}$ stochastically dominates $D_{\mathrm{bm}}$.\footnote{Given two random variables $A$ and $B$, we say that $A$ stochastically dominates $B$ if $\mathbb{P}(A > x) \geq \mathbb{P}(B > x)$  for all $a$, or alternatively, $F_A (x) \leq F_B (x)$ for cumulative distribution functions $F_A (x)$ and $F_B (x)$.} This implies that small cell backhaul links are subject to higher delays compared to those of \glspl{MBS}.

%%%%%%%%%%%%%%%%%%%%%%%%%%%%%%%%%%% Delay in caching
\emph{Delay in caches}: Serving a user by fetching its content from the local cache is subject to delay as the storage medium is prone to errors, whereas such a delay may also vary depending on the storage type and the underlying mechanisms (i.e., hard disk, \ac{SSD}). In this regard, we model this phenomenon as $D_{\mathrm{ca}} \sim \mathrm{Exponential}(\bar{\mu}_{\mathrm{ca}})$, an exponentially distributed random variable with mean $\bar{\mu}_{\mathrm{ca}}$ being proportional to the storage type. We also assume that the delay of small cell backhaul links stochastically dominates the delay of reading a content from local caches, meaning that the speed of content reads from caches is stochastically higher than the speed of small cell backhaul links.
%%%%%%%%%%%%%%%%%%%%%%%%%%%%%%%%%%%%%%%%%%%%%%%%%%%%%%%%%%%%%%%%%%%%
%%%%%%%%%%%%%%%%%%%%%%%%%%%%%%%%%%% PERFORMANCE ANALYSIS
\section{Performance Analysis}\label{sec:performance}
%%%%%%%%%%%%%%%%%%%%%%%%%%%%%%%%%%%%%%%%%%%%%%%%%%%%%%%%%%%%%%%%%%%%
%%%%%%%%%%%%%%%%%%%%%%%%%%%%%%%%%%% DELAY ANALYSIS
%\subsection{Delay Analysis}
Considering the aforementioned sources of delay, namely downlink, caching and backhaul, the delay experienced by the typical \glspl{MU} and \glspl{SU} are respectively defined as
\begin{align}
	D_{\mathrm{m}} &= D_{\mathrm{dm}} + D_{\mathrm{bm}}, \label{eq:delaymc} \\
	D_{\mathrm{s}} &= D_{\mathrm{ds}}
					 + \mathds{1}_{\{ f_{\mathrm{s}} \in \Delta_0 \}}D_{\mathrm{ca}} 
					 + \big(1 - \mathds{1}_{\{ f_{\mathrm{s}} \in \Delta_0 \}} \big) D_{\mathrm{bs}} \label{eq:delaysc}
\end{align}
where $f_{\mathrm{s}}$ is the content requested by the typical small cell user and $\Delta_0$ is the cache of its associated small cell. The indicator function $\mathds{1}_{\{ ... \}}$ returns $1$ if the statement holds, $0$ otherwise. Before proceeding to the next step, let us define functions $B_1(T_0, M, \gamma, \alpha, P_x, P_y, \lambda_x,  \lambda_y )$, $B_2(S_{\mathrm{p}}, \eta)$ and $B_3(S_{\mathrm{u}}, S_{\mathrm{p}}, f_0, \eta)$ given on the top of the next page. 
%%%%%%%%%%%%%%%%%%%%% Common Functions: B_1 B_2 B_3
\begin{figure*}[!t]
{\small
\begin{align}
B_1(T_0, M, \gamma, \alpha, P_x, P_y, \lambda_x,  \lambda_y )
&=
	T_0
	\sum_{i=0}^{M-1}(-1)^i \binom{M}{i + 1} 
	\frac{1}
	{1 +
		i\big[
				\rho(\gamma, \alpha) +
				(P_x/P_y)^{2/\alpha}
				(\lambda_x/\lambda_y)
				\gamma^{2/\alpha}
				A(\alpha)
		\big]	
	} \label{eq:b1} \\
B_2(S_{\mathrm{p}}, \eta)
&= 1 - \big(1 + S_{\mathrm{p}} \big)^{1 - \eta} \label{eq:b2}  \\
B_3(S_{\mathrm{u}}, S_{\mathrm{p}}, f_0, \eta)  
&= 
	\frac{S_{\mathrm{u}}}{f_0 - S_{\mathrm{p}}}
	\Big(
		1 - \big(1 + f_0 \big)^{1 - \eta} + \big(1 + S_{\mathrm{p}} \big)^{1 - \eta}
	\Big)
	\label{eq:b3} 
\end{align}
}
\vspace{-0.3cm}
\end{figure*}

We now state the following result related to the average delay experienced by the typical \glspl{MU}.
%%%%%%%%%%%%%%%%%%%%% Theorem - Average Downlink Delay
\begin{theorem}\label{theor:avgdelaymc}
The average delay for a typical user connected to its nearest \ac{MBS} is given by
\begin{align}
	\bar{D}_{\mathrm{m}} = B_1(T_0, M, \gamma, \alpha, P_{\mathrm{sc}}, P_{\mathrm{mc}}, \lambda_{\mathrm{sc}},  \lambda_{\mathrm{mc}} ) 	+
						   \frac{1}{2} \beta \lambda_{\mathrm{mc}} \lambda_{\mathrm{cs}}^{-3/2}
\end{align}
where $B_1(T_0, M, \gamma, \alpha, P_{\mathrm{sc}}, P_{\mathrm{mc}}, \lambda_{\mathrm{sc}},  \lambda_{\mathrm{mc}} )$ is given in \eqref{eq:b1}. The parameter $\beta$ is a scaling factor, relating to the importance of backhaul delay over the non-backhaul delay.
\end{theorem}
\begin{proof}
See Appendix B.2 in \cite{Bastug2015Distributed}.
\end{proof}

In Theorem \ref{theor:avgdelaymc}, the function $B_1$ models the average downlink delay whereas the remaining term in $\bar{D}_{\mathrm{m}}$ incorporates the average delay caused due to the backhaul. The summation of terms is due to the consideration of independent \glspl{PPP}. We now turn our attention to \ac{SU} with and without caching capabilities at the \glspl{SBS}.
\begin{corollary}\label{theor:avgdelaysc:nocache}
The average delay for a typical user connected to its nearest small cell (with no caching) is given by
{\small
\begin{align}
	\bar{D}_{\mathrm{m}} = B_1(T_0, M, \gamma, \alpha, P_{\mathrm{mc}}, P_{\mathrm{sc}}, \lambda_{\mathrm{mc}},  \lambda_{\mathrm{sc}} ) 	+
						   \frac{1}{2} \beta \lambda_{\mathrm{sc}} \lambda_{\mathrm{cs}}^{-3/2}
\end{align}
}
where $B_1(T_0, M, \gamma, \alpha, P_{\mathrm{mc}}, P_{\mathrm{sc}}, \lambda_{\mathrm{mc}},  \lambda_{\mathrm{sc}} )$ is given in \eqref{eq:b1}.
\end{corollary}
\begin{proof}
The result is a direct application of Theorem \ref{theor:avgdelaymc}, thus it is immediately proved by following similar steps given in Appendix B.2 in \cite{Bastug2015Distributed}.
\end{proof}
%%%%%%%%%%%%%%%%%%%%% Theorem - Average Downlink Delay (MixPop)
\begin{theorem}\label{theor:avgdelaysc:mixpop}
When $\mathrm{MixPop}$ caching policy is employed at the \glspl{SBS}, the average delay for a typical user connected to its nearest small cell under fixed content popularity distribution is given by
{\small
\begin{multline}\label{theor:avgdelaysc:mixpopFixed}
	\bar{D}_{\mathrm{fix}}^{\mathrm{(mix)}} =
		B_1(T_0, M, \gamma, \alpha, P_{\mathrm{mc}}, P_{\mathrm{sc}}, \lambda_{\mathrm{mc}},  \lambda_{\mathrm{sc}} ) +
		\frac{1}{2} \beta \lambda_{\mathrm{sc}} \lambda_{\mathrm{cs}}^{-3/2} +   \\
		\big( \bar{\mu}_{\mathrm{ca}} - \frac{1}{2} \beta \lambda_{\mathrm{sc}} \lambda_{\mathrm{cs}}^{-3/2} \big)
		\Big( 
			B_2(S_{\mathrm{p}}, \eta_0) +
			B_3(S_{\mathrm{u}}, S_{\mathrm{p}}, f_0, \eta_0)
		\Big) 		
\end{multline}
}
where $B_2(S_{\mathrm{p}}, \eta_0)$ and $B_3(S_{\mathrm{u}}, S_{\mathrm{p}}, f_0, \eta_0)$ are given in \eqref{eq:b2} and \eqref{eq:b3} respectively.

In case of distance-dependent content popularity, the average delay is given by
{\small
\begin{multline}\label{theor:avgdelaysc:mixpopDistance}
	\bar{D}_{\mathrm{dist}}^{\mathrm{(mix)}} =
		B_1(T_0, M, \gamma, \alpha, P_{\mathrm{mc}}, P_{\mathrm{sc}}, \lambda_{\mathrm{mc}},  \lambda_{\mathrm{sc}} ) +
		\frac{1}{2} \beta \lambda_{\mathrm{sc}} \lambda_{\mathrm{cs}}^{-3/2} +   \\
		\big( \bar{\mu}_{\mathrm{ca}} - \frac{1}{2} \beta \lambda_{\mathrm{sc}} \lambda_{\mathrm{cs}}^{-3/2} \big)
		\Big( 
			B_2(S_{\mathrm{p}}, \frac{1}{2\sqrt{\lambda_{\mathrm{sc}}}}) + \\
			B_3(S_{\mathrm{u}}, S_{\mathrm{p}}, f_0, \frac{1}{2\sqrt{\lambda_{\mathrm{sc}}}})
		\Big) 		
		.
\end{multline}
}

In case of load-dependent content popularity, the average delay is given by
{\small
\begin{multline}\label{theor:avgdelaysc:mixpopLoad}
	\bar{D}_{\mathrm{load}}^{\mathrm{(mix)}} =
		B_1(T_0, M, \gamma, \alpha, P_{\mathrm{mc}}, P_{\mathrm{sc}}, \lambda_{\mathrm{mc}},  \lambda_{\mathrm{sc}} ) +
		\frac{1}{2} \beta \lambda_{\mathrm{sc}} \lambda_{\mathrm{cs}}^{-3/2} +   \\
		\big( \bar{\mu}_{\mathrm{ca}} - \frac{1}{2} \beta \lambda_{\mathrm{sc}} \lambda_{\mathrm{cs}}^{-3/2} \big)
		\Big( 
			B_2(S_{\mathrm{p}}, \frac{\lambda_{\mathrm{ut}}}{\lambda_{\mathrm{sc}}} ) + \\
			B_3(S_{\mathrm{u}}, S_{\mathrm{p}}, f_0, \frac{\lambda_{\mathrm{ut}}}{\lambda_{\mathrm{sc}}} )
		\Big) 	
		.
\end{multline}
}
\end{theorem}
\begin{proof}
See Appendix B.3 in \cite{Bastug2015Distributed}.
\end{proof}
The functions $B_2$ and $B_3$ in Theorem \ref{theor:avgdelaysc:mixpop} are related to caching popular contents and caching uniformly at random respectively, and captures the cache hit behavior of the MixPop policy. By slightly modifying the steps in the proof of Theorem \ref{theor:avgdelaysc:mixpop}, similar results for StdPop and UniRand caching policies can be readily obtained. Note that the results above are based on the assumption that the typical users are connected to their nearest base stations. 

In the above, we have provided the average delay expressions for typical \glspl{MU} and \glspl{SU}. The total average network delay, total network cost (including deployment and operational costs), and optimization of these metrics with respect to system design parameters are left for future work. 
%%%%%%%%%%%%%%%%%%%%%%%%%%%%%%%%%%%%%%%%%%%%%%%%%%%%%%%%%%%%%%%%%%%%
%%%%%%%%%%%%%%%%%%%%%%%%%%%%%%%%%%% NUMERICAL RESULTS
\section{Numerical Results}\label{sec:numresults}
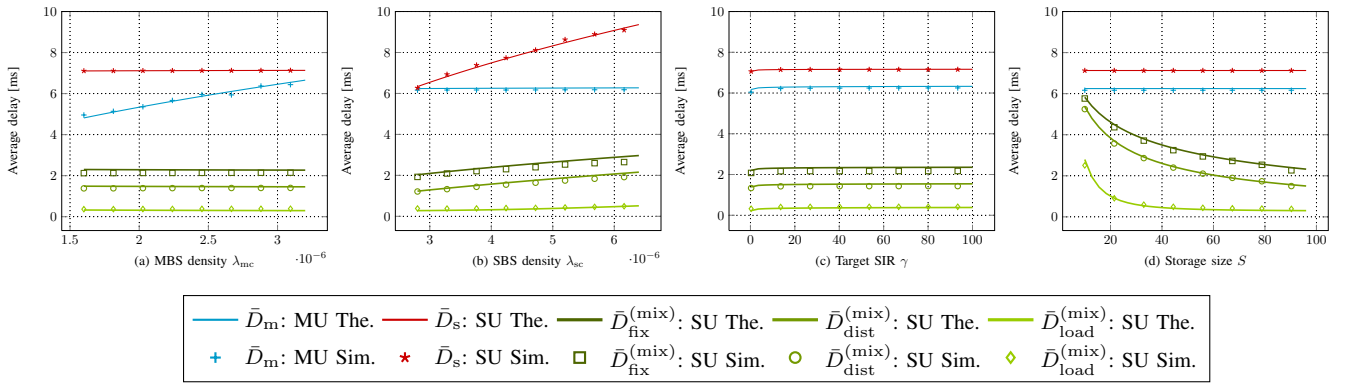
\begin{figure*}[ht!]
\centering
%\begin{subfigure}[t]{0.35\textwidth}
\begin{tikzpicture}[scale=0.51, baseline]
	\begin{axis}[
    every tick label/.append style  =
    { 
        font=\normalsize
    },
		ymax = 10,
 		grid = major,
 		cycle list name={laneas-delay1},
% 		legend cell align=left,
 		mark repeat={2},
	  legend columns=5,
		legend entries={
			$\bar{D}_{\mathrm{m}}$: MU The.,
		  $\bar{D}_{\mathrm{s}}$: SU The.,
			$\bar{D}_{\mathrm{fix}}^{\mathrm{(mix)}}$: SU The.,
		  $\bar{D}_{\mathrm{dist}}^{\mathrm{(mix)}}$: SU The.,
			$\bar{D}_{\mathrm{load}}^{\mathrm{(mix)}}$: SU The.,
			$\bar{D}_{\mathrm{m}}$: MU Sim.,
			$\bar{D}_{\mathrm{s}}$: SU Sim.,
			$\bar{D}_{\mathrm{fix}}^{\mathrm{(mix)}}$: SU Sim.,
			$\bar{D}_{\mathrm{dist}}^{\mathrm{(mix)}}$: SU Sim.,
			$\bar{D}_{\mathrm{load}}^{\mathrm{(mix)}}$: SU Sim.
		},
		legend cell align=left,
	  legend to name=namedmethods4stogeodelay,
		label style={font=\normalsize},
 		xlabel={{\normalsize (a) MBS density $\lambda_{\mathrm{mc}}$}},
 		ylabel={Average delay [ms]}
		]
 		%%%%%%%%%%%%%%% Theory
 		% MU
 		\addplot+table [col sep=comma] {\string"intensityMc-theoryAvgTotalDelayMu.csv"};
		% SU - No Caching
 		\addplot+table [col sep=comma] {\string"intensityMc-theoryAvgTotalDelaySuNoCache.csv"};
		% SU - MixPop + Fixed	  
 		\addplot+ table [col sep=comma] {\string"intensityMc-theoryAvgTotalDelaySuMixFixed.csv"};
		% SU - MixPop + Distance 
 		\addplot+ table [col sep=comma] {\string"intensityMc-theoryAvgTotalDelaySuMixDistance.csv"};
		% SU - MixPop + Load  
 		\addplot+ table [col sep=comma] {\string"intensityMc-theoryAvgTotalDelaySuMixLoad.csv"};
		
		%%%%%%%%%%%%%%%  Simulation		  
		% MU			
 		\addplot+table [col sep=comma]{\string"intensityMc-simAvgTotalDelayMu.csv"};			
		% SU - No Caching
 		\addplot+table [col sep=comma] {\string"intensityMc-simAvgTotalDelaySuNoCache.csv"};	
		% SU - MixPop + Fixed	
 		\addplot+ table [col sep=comma] {\string"intensityMc-simAvgTotalDelaySuMixFixed.csv"};			
		% SU - MixPop + Distance
 		\addplot+ table [col sep=comma] {\string"intensityMc-simAvgTotalDelaySuMixDistance.csv"};			
		% SU - MixPop + Load	
 		\addplot+ table [col sep=comma] {\string"intensityMc-simAvgTotalDelaySuMixLoad.csv"};				  		
	\end{axis}
\end{tikzpicture}
%\end{subfigure}
%%%%%%%%%%%%%%%%%%%%%%%%%%%%%%%%%%%%%%%%%%%%%%%%%% Impact of small cell density
%\begin{subfigure}[t]{0.35\textwidth}
\begin{tikzpicture}[scale=0.51, baseline]
	\begin{axis}[
    every tick label/.append style  =
    { 
        font=\normalsize
    },
		ymax = 10,
 		grid = major,
 		cycle list name={laneas-delay1},
% 		legend cell align=left,
 		mark repeat={2},
		label style={font=\normalsize},
 		xlabel={{\normalsize (b) SBS density $\lambda_{\mathrm{sc}}$}},
 		ylabel={Average delay [ms]}
		]
 		%%%%%%%%%%%%%%% Theory
 		% MU
 		\addplot+table [col sep=comma] {\string"intensitySc-theoryAvgTotalDelayMu.csv"};
		% SU - No Caching
 		\addplot+table [col sep=comma] {\string"intensitySc-theoryAvgTotalDelaySuNoCache.csv"};
		% SU - MixPop + Fixed	  
 		\addplot+ table [col sep=comma] {\string"intensitySc-theoryAvgTotalDelaySuMixFixed.csv"};
		% SU - MixPop + Distance 
 		\addplot+ table [col sep=comma] {\string"intensitySc-theoryAvgTotalDelaySuMixDistance.csv"};
		% SU - MixPop + Load  
 		\addplot+ table [col sep=comma] {\string"intensitySc-theoryAvgTotalDelaySuMixLoad.csv"};
		
		%%%%%%%%%%%%%%%  Simulation		  
		% MU			
 		\addplot+table [col sep=comma]{\string"intensitySc-simAvgTotalDelayMu.csv"};			
		% SU - No Caching
 		\addplot+table [col sep=comma] {\string"intensitySc-simAvgTotalDelaySuNoCache.csv"};	
		% SU - MixPop + Fixed	
 		\addplot+ table [col sep=comma] {\string"intensitySc-simAvgTotalDelaySuMixFixed.csv"};			
		% SU - MixPop + Distance
 		\addplot+ table [col sep=comma] {\string"intensitySc-simAvgTotalDelaySuMixDistance.csv"};			
		% SU - MixPop + Load	
 		\addplot+ table [col sep=comma] {\string"intensitySc-simAvgTotalDelaySuMixLoad.csv"};							  			  		
	\end{axis}
\end{tikzpicture}
%\end{subfigure}
%%%%%%%%%%%%%%%%%%%%%%%%%%%%%%%%%%%%%%%%%%%%%%%%%% Impact of target SIR
%\begin{subfigure}[t]{0.35\textwidth}
\begin{tikzpicture}[scale=0.51, baseline]
	\begin{axis}[
    every tick label/.append style  =
    { 
        font=\normalsize
    },
		ymax = 10,
 		grid = major,
 		cycle list name={laneas-delay1},
% 		legend cell align=left,
 		mark repeat={2},
		label style={font=\normalsize},
 		xlabel={{\normalsize (c) Target SIR $\gamma$}},
 		ylabel={Average delay [ms]}
		]
 		%%%%%%%%%%%%%%% Theory
 		% MU
 		\addplot+table [col sep=comma] {\string"targetSir-theoryAvgTotalDelayMu.csv"};
		% SU - No Caching
 		\addplot+table [col sep=comma] {\string"targetSir-theoryAvgTotalDelaySuNoCache.csv"};
		% SU - MixPop + Fixed	  
 		\addplot+ table [col sep=comma] {\string"targetSir-theoryAvgTotalDelaySuMixFixed.csv"};
		% SU - MixPop + Distance 
 		\addplot+ table [col sep=comma] {\string"targetSir-theoryAvgTotalDelaySuMixDistance.csv"};
		% SU - MixPop + Load  
 		\addplot+ table [col sep=comma] {\string"targetSir-theoryAvgTotalDelaySuMixLoad.csv"};
		
		%%%%%%%%%%%%%%%  Simulation		  
		% MU			
 		\addplot+table [col sep=comma]{\string"targetSir-simAvgTotalDelayMu.csv"};			
		% SU - No Caching
 		\addplot+table [col sep=comma] {\string"targetSir-simAvgTotalDelaySuNoCache.csv"};	
		% SU - MixPop + Fixed	
 		\addplot+ table [col sep=comma] {\string"targetSir-simAvgTotalDelaySuMixFixed.csv"};			
		% SU - MixPop + Distance
 		\addplot+ table [col sep=comma] {\string"targetSir-simAvgTotalDelaySuMixDistance.csv"};			
		% SU - MixPop + Load	
 		\addplot+ table [col sep=comma] {\string"targetSir-simAvgTotalDelaySuMixLoad.csv"};								  			  		
	\end{axis}
\end{tikzpicture}
%\end{subfigure}
%%%%%%%%%%%%%%%%%%%%%%%%%%%%%%%%%%%%%%%%%%%%%%%%%% Impact of storage size
%\begin{subfigure}[t]{0.35\textwidth}
\begin{tikzpicture}[scale=0.51, baseline]
	\begin{axis}[
    every tick label/.append style  =
    { 
        font=\normalsize
    },
		ymax = 10,
 		grid = major,
 		cycle list name={laneas-delay1},
% 		legend cell align=left,
 		mark repeat={2},
		label style={font=\normalsize},
 		xlabel={{\normalsize (d) Storage size $S$}},
 		ylabel={Average delay [ms]}
		]
 		%%%%%%%%%%%%%%% Theory
 		% MU
 		\addplot+table [col sep=comma] {\string"storageSz-theoryAvgTotalDelayMu.csv"};
		% SU - No Caching
 		\addplot+table [col sep=comma] {\string"storageSz-theoryAvgTotalDelaySuNoCache.csv"};
		% SU - MixPop + Fixed	  
 		\addplot+ table [col sep=comma] {\string"storageSz-theoryAvgTotalDelaySuMixFixed.csv"};
		% SU - MixPop + Distance 
 		\addplot+ table [col sep=comma] {\string"storageSz-theoryAvgTotalDelaySuMixDistance.csv"};
		% SU - MixPop + Load  
 		\addplot+ table [col sep=comma] {\string"storageSz-theoryAvgTotalDelaySuMixLoad.csv"};
		
		%%%%%%%%%%%%%%%  Simulation		  
		% MU			
 		\addplot+table [col sep=comma]{\string"storageSz-simAvgTotalDelayMu.csv"};			
		% SU - No Caching
 		\addplot+table [col sep=comma] {\string"storageSz-simAvgTotalDelaySuNoCache.csv"};	
		% SU - MixPop + Fixed	
 		\addplot+ table [col sep=comma] {\string"storageSz-simAvgTotalDelaySuMixFixed.csv"};			
		% SU - MixPop + Distance
 		\addplot+ table [col sep=comma] {\string"storageSz-simAvgTotalDelaySuMixDistance.csv"};			
		% SU - MixPop + Load	
 		\addplot+ table [col sep=comma] {\string"storageSz-simAvgTotalDelaySuMixLoad.csv"};				  			  		
	\end{axis}
\end{tikzpicture}
%\end{subfigure}
%%%%%%%%%%%%%%%%%%%%%%%%%%%%%%%%%%%%%%%%%%%%%%%%%% Caption and legends
\vspace*{0.3cm}
\\
\normalsize
\hspace{0.65cm}\ref{namedmethods4stogeodelay}
\caption{Evolution of average delay with respect to the a) macro cell density, b) small cell density, c) target SIR and d) storage size. $\lambda_{\mathrm{cr}} = 1.4\times 10^{-6}$, 
$\lambda_{\mathrm{mc}} = 2.8\times 10^{-6}$, 
$\lambda_{\mathrm{sc}} = 3.6\times 10^{-6}$, 
$\lambda_{\mathrm{ut}} = 7.2\times 10^{-6}$;
$P_{\mathrm{mc}} = 20$,
$P_{\mathrm{sc}} = 2$ Watts;
$\alpha = 4$;
$\gamma = 3$ dB;
$M = 4$;
$T_0 = 0.1$,
$\mu_{\mathrm{ca}} = 0.01$ ms;
$\eta_0 = 1.45$;
$f_0 = 500$, 
$S = 100$,
$S_{\mathrm{p}} = 9.5$,
$S_{0} = 0.5$,
$S_{\mathrm{u}} = 90$ GByte.}
\label{fig:avgDelay:main}
\vspace{-0.3cm}
\end{figure*}

In this section, we numerically validate our approximations derived in the previous section. The impact of critical system parameters are discussed as follows. 

%%%%%%%%%%%%%%%%%%%%%%%%%%%%%%%%%%%
{\bf Impact of \ac{MBS} density} $\lambda_{\mathrm{mc}}$: The change of average delay with respect to the \ac{MBS} density is given in Fig. \ref{fig:avgDelay:main}a. Therein, as the number of \glspl{MBS} increases, we observe an increment in average delay. This is mainly due to the backhaul as the delay in backhaul is proportional to the distance and average number of connected \glspl{MBS}. In this setup, even though the average distance from a \ac{MBS} to its central central router decreases (thus less delay in the backhaul), the increasing number of base stations contributes more to the average delay, thus yielding such a behaviour. On the other hand, the average delay in \glspl{SBS} remains static in this setup. However, we note that the average delay experienced by a typical small cell user is reduced by adding caching capabilities at the base stations. For instance, when content popularity is load-dependent and caching policy is MixPop, the average delay is reasonably less than other candidates (including typical users with no caching at \glspl{SBS}).

%%%%%%%%%%%%%%%%%%%%%%%%%%%%%%%%%%%
{\bf Impact of small cell density} $\lambda_{\mathrm{sc}}$: The change of the average delay with respect to the small cell density is depicted in Fig. \ref{fig:avgDelay:main}b. Similarly to the previous figure for \ac{MBS} density, we see that the average delay increases for all kind of small cell users. However, in this numerical setup, the rate of increment in delay with no-caching capabilities at the \glspl{SBS} is higher than the delay experienced by the typical users with cache-enabled \glspl{SBS}. Compared to the fixed and load-dependent content popularities, the typical user under load-dependent content popularity experiences less delay when the number of \glspl{SBS} increases.

{\bf Impact of target \ac{SIR}} $\gamma$: In our setup, yet another important design parameter is the target \ac{SIR}. In this regard, the average delay variation with respect to the target \ac{SIR} is illustrated in Fig. \ref{fig:avgDelay:main}c. As observed in the figure, the average delay increases by imposing higher target \ac{SIR} values. This change is only visible in low values of target \ac{SIR}, whereas the variation of delay in higher values of target \ac{SIR} is negligible. This might stem from the fact that the downlink delay is not a dominating factor in our scenario compared to the backhaul delay. A typical user connected to the small cell with no caching capabilities experiences the highest delay, whereas the minimum delay is achieved by using MixPop policy under load-dependent content popularity. The delay of a typical \ac{MU} remains between a \ac{SU} with no-caching and caching capabilities at the base stations.

{\bf Impact of storage size} $S$: Yet another crucial design parameter in our setup is the storage size. The impact of storage size on the average delay is shown in Fig. \ref{fig:avgDelay:main}d.  Indeed, as observed from the figure, dramatical decrease in delay is observed by increasing the storage size of small base stations. Similarly to previous observations, the most sensitive content popularity for the average delay is the load-dependent content popularity.
%%%%%%%%%%%%%%%%%%%%%%%%%%%%%%%%%%%%%%%%%%%%%%%%%%%%%%%%%%%%%%%%%%%%
%%%%%%%%%%%%%%%%%%%%%%%%%%%%%%%%%%% CONCLUSIONS and FUTURE WORK
\vspace{-0.30cm}
\section{Conclusions}\label{sec:confut}
In this work, we have characterized the average delay of \glspl{MU} and \glspl{SU} under backhaul constraints and caching capabilities at the small base stations. Several content popularity distributions and caching policies have been considered. The main conclusion from this work is that caching at the small base stations allows for balancing the average access delay to the contents, especially if heterogeneous network densification under limited backhaul is considered. 
%%%%%%%%%%%%%%%%%%%%%%%%%%%%%%%%%%%%%%%%%%%%%%%%%%%%%%%%%%%%%%%%%%%
%%%%%%%%%%%%%%%%%%%%%%%%%%%%%%%%%% BIBLIOGRAPHY
\bibliographystyle{IEEEtran}
\bibliography{references}
\end{document}